\documentclass[11pt,a4paper,reqno]{amsart}
\pdfoutput=1

\usepackage[T1]{fontenc}
\usepackage[utf8]{inputenc}
\usepackage[a4paper,margin=1in]{geometry}
\usepackage{lmodern}
\usepackage{amssymb}
\usepackage{booktabs}
\usepackage{microtype}
\usepackage{tikz}
\usetikzlibrary{decorations.pathmorphing}
\usepackage[hidelinks]{hyperref}

\usepackage{orcidlink}
\usepackage[nameinlink,noabbrev]{cleveref}
\usepackage{aliascnt}

\newcommand{\Tr}{\operatorname{Tr}}
\newcommand{\supp}{\operatorname{supp}}
\newcommand{\rank}{\operatorname{rank}}
\newcommand{\id}{\operatorname{id}}

\newcommand{\diam}{\diamond}
\newcommand{\ketbra}[2]{\lvert #1\rangle\langle #2\rvert}
\newcommand{\abs}[1]{\lvert #1\rvert}
\newcommand{\norm}[1]{\lVert #1\rVert}
\newcommand{\inner}[2]{\langle #1,#2\rangle}

\theoremstyle{plain}
\newtheorem{theorem}{Theorem}
\newaliascnt{proposition}{theorem}
\newtheorem{proposition}[proposition]{Proposition}
\aliascntresetthe{proposition}
\crefname{proposition}{proposition}{propositions}
\newaliascnt{lemma}{theorem}
\newtheorem{lemma}[lemma]{Lemma}
\aliascntresetthe{lemma}
\crefname{lemma}{lemma}{lemmas}

\theoremstyle{definition}

\title{Robust logarithmic entanglement lower bound for $f$-routing}
\author[Kevin Bogner]{Kevin Bogner\,\orcidlink{0009-0003-5167-7177}\\
 \NoCaseChange{{\normalfont\small COSIC, KU Leuven, Belgium}}\\
 \NoCaseChange{{\normalfont\small\texttt{kevin.bogner@kuleuven.be}}}}

\begin{document}

\begin{abstract}
One-round $f$-routing models attacks on a quantum position verification protocol. In this task, Alice receives an $n$-bit string $x$ and an unknown qubit, and Bob receives an $n$-bit string $y$. They share a state prepared in advance and exchange one simultaneous message each. The party selected by the Boolean function $f(x,y)$ must then recover the qubit without further communication. For inner product modulo $2$, this paper proves that, for each constant $\epsilon<0.24$, every protocol with error at most $\epsilon$ on average over uniformly random inputs needs a shared state whose entanglement of formation grows at least logarithmically in $n$. The bound allows error both when Alice must recover the qubit and when Bob must. It also allows arbitrary mixed shared states and places no limit on message length or local operations. In this model, the other known growing bounds for explicit functions that charge only the shared state require perfect recovery for at least one of the two parties. Whether a robust polynomial lower bound on the shared entanglement holds remains open.
\end{abstract}

\maketitle
\enlargethispage{-\baselineskip}

\section{Introduction}

Position verification aims to certify where a prover is located using only timed messages. In the simplest setting, two verifiers sit on a line, one on each side of the claimed position. The left verifier sends an $n$-bit string $x$, the right verifier sends an $n$-bit string $y$, and both expect the outcome of $f(x,y)$ back in time, for a fixed Boolean function $f$. A prover at the claimed position receives both strings together and can answer both verifiers in time. A single prover anywhere else cannot receive both strings and still answer both verifiers in time. Whether it can predict $f(x,y)$ without them depends on the function and the randomness of the inputs. Two colluding attackers nevertheless defeat this scheme for every $f$, one standing between the claimed position and each verifier. Each intercepts the nearer string, keeps a copy, and forwards another copy to the partner, so both can compute $f(x,y)$ and answer in time. No classical scheme resists such colluding attackers~\cite{cgmo}.

This attack copies its inputs, and the no-cloning theorem~\cite{wootterszurek} forbids copying an unknown quantum state. A routing scheme therefore adds an unknown qubit, which the left verifier sends together with $x$, and asks the prover to return it to the verifier selected by $f(x,y)$~\cite{tagging,gardenhose}. An attacker cannot both keep the qubit and forward it, so the copying attack fails~\cite{nogo}. Attackers who share entanglement before the protocol starts can still move quantum inputs between them by quantum teleportation within one round of simultaneous communication and consumption of the shared entanglement. In the standard model, attackers may apply arbitrary local quantum operations and exchange quantum messages. Attackers with enough shared entanglement then break every position-verification scheme with small error~\cite{nogo}. Such generic attacks consume exponentially many EPR pairs, the maximally entangled two-qubit states~\cite{nogo,portbased}. How much entanglement attackers need was left open~\cite{nogo}. An honest prover computes $f(x,y)$ classically and then routes a single qubit, so its quantum cost does not grow with $n$. If attacks need shared entanglement that grows with $n$, the verifiers can therefore choose $n$ large enough to make the routing scheme secure against attackers with any fixed amount of entanglement~\cite{gardenhose}.

For the routing scheme, the attackers' task is one-round $f$-routing~\cite{gardenhose,relating}. Alice, the attacker near the left verifier, receives $x\in\{0,1\}^n$ and the qubit $Q$, while Bob receives $y\in\{0,1\}^n$. Before the inputs arrive, they share a state $\rho_{LR}$, with Alice holding $L$ and Bob holding $R$. They exchange one simultaneous message each and then recover $Q$ locally, Alice when $f(x,y)=0$ and Bob when $f(x,y)=1$. These two possibilities are the routing cases. In the standard model for $f$-routing~\cite{maynotes,complexity}, recovery must succeed, up to a small error, on every input pair $(x,y)$, even the worst one, and not only on most pairs. The error compares the original qubit with the qubit that the party selected by $f(x,y)$ outputs, also when the original is entangled with another system that neither attacker holds. If the qubit ends up with the other party, the selected party does not have it, so the error is large. This paper uses the same model, but its lower bound assumes a small error only on average over uniformly random input pairs, so it also holds for protocols that fail on some pairs. Lower bounds on the shared state in this model therefore constrain attacks on the routing scheme~\cite{relating,complexity}.

Intuitively, the need for entanglement comes from timing. Alice must act on $Q$ before she learns $y$, and she cannot both keep $Q$ and send it. If $f$ depends only on $x$, Alice can decide alone and needs no entanglement. If $f$ depends only on $y$, one shared EPR pair suffices. Alice teleports $Q$ into Bob's half of the pair, keeps a copy of her two outcome bits, and sends another copy to Bob. Bob keeps his half or sends it to Alice according to $f$, and the party who ends up with it applies the correction. When $f$ depends on both inputs, neither party can make this decision alone.

Honest provers are noisy, so the verifiers must tolerate a small error in both routing cases. Attackers then only need to reach that error, so a lower bound that matters for security must hold for protocols that are noisy in both routing cases~\cite{cdsrouting}. This paper calls such a bound robust. Previous growing bounds for explicit functions in this model that charge only the shared state require perfect recovery in at least one routing case~\cite{ranklower,cdsrouting}. Their method uses a matrix indexed by the input pairs whose entries vanish on the inputs of the perfectly recovered routing case. With error in both routing cases, no entry is forced to vanish, and the method gives no bound.

\subsection{Inner-product routing needs growing entanglement}

For inner product modulo $2$, this paper proves that, for each constant $\epsilon<0.24$, every protocol with average error at most $\epsilon$ needs a shared state whose entanglement of formation grows at least logarithmically in $n$, even if the state is mixed.

The entanglement of formation~\cite{bdsw} measures how much entanglement a shared state contains. For a pure state, it is the entropy of either party's local state, which is zero for a separable pure state, and one for an EPR pair. For a mixed state, local entropy can also come from classical correlation, so the entanglement of formation takes the smallest average entropy over all ways to prepare the state by drawing a pure state $\psi_i$ with probability $p_i$, which gives $\rho_{LR}=\sum_ip_i\ketbra{\psi_i}{\psi_i}$,
\begin{equation}
 E_F(\rho_{LR})=\min_{\{p_i,\psi_i\}}\sum_ip_i\,S((\psi_i)_L),
\end{equation}
where $S$ is the von Neumann entropy in bits. The entanglement of formation vanishes exactly on separable states, which Alice and Bob can prepare from shared classical randomness and local operations alone, so it charges entanglement and not classical correlation. This makes it a suitable measure of the entanglement that the attackers must share in advance to break the protocol.

The main result of this paper is a lower bound on the entanglement of formation, stated with the test of the routing scheme in~\cite{gardenhose}. There the unknown qubit is one half of an EPR pair, and the verifiers measure whether the returned qubit and the other half still form this pair. Let $a$ be the probability that this test catches the attackers, on average over uniformly random inputs. \Cref{sec:model} defines this probability formally and calls it the average entanglement infidelity.

\begin{theorem}[Robust entanglement lower bound]\label{thm:formation}
For all $n\ge2$, every protocol for
$f_n(x,y)=\bigoplus_{i=1}^n x_iy_i$
in the model of \Cref{sec:model} with average entanglement infidelity $a<0.12$ satisfies
\begin{equation}
 E_F(\rho_{LR})
 \ge(0.12-a)^2
 \left(\log_2n-\log_2\log_2n-9\right),
\end{equation}
where $\rho_{LR}$ is the initially shared state.
\end{theorem}

In words, attacks on the routing scheme for inner product modulo $2$ need a shared state whose entanglement grows at least logarithmically in $n$, independently of the length of their messages and of their local operations. The factor $(0.12-a)^2$ shrinks as the probability $a$ of being caught approaches $0.12$, but the lower bound holds as long as $a$ is below $0.12$. For a protocol with average error at most $\epsilon<0.24$, the same bound holds with coefficient $(0.12-\epsilon/2)^2$, since $a\le\epsilon/2$ by \Cref{sec:model}.

Earlier lower bounds that charge only the shared state measure its dimension, through the logarithm of the Schmidt rank of a pure shared state~\cite{ranklower,cdsrouting}. Related work also uses the logarithm of the smaller local dimension $\min\{\dim L,\dim R\}$~\cite{maynotes}. Such dimension costs can be large while the shared state has little entanglement. Every pure component has Schmidt rank at most the smaller local dimension, so \Cref{lem:bounded-rank-ensemble} gives the same lower bound for both costs without the factor $(0.12-a)^2$.

A garden-hose protocol routes inner product perfectly with $4n+1$ shared EPR pairs~\cite{gardenhose}. With an exponentially long quantum message from Alice, $2n+2$ shared EPR pairs suffice to route every Boolean function perfectly~\cite{linearupper}. Whether a robust lower bound on the shared entanglement can grow faster than logarithmically, and in particular polynomially as asked in~\cite{maynotes}, remains open.

\subsection{Related bounds restrict the error, the attackers or the charged quantity}

\Cref{tab:comparison} compares the resources charged by each lower bound and the errors it allows.

\begin{table}[t]
\centering
\caption{Lower bounds for one-round $f$-routing that grow with the length $n$ of the classical inputs. Here $\mathrm{SR}(\psi)$ is the Schmidt rank of a pure shared state. The qubit count includes each attacker's local memory and messages. The gate count in~\cite{gatebounds} includes single-qubit measurements, and $q$ is the number of qubits the attackers hold. The errors $\epsilon_0$ and $\epsilon_1$ apply when Alice and when Bob must recover, and $c$ is a constant. The error bound holds on every input ($\forall$) or on average over uniform inputs ($\mathbb E$).}\label{tab:comparison}
\begin{tabular}{lllcll}
\toprule
Work & Charged quantity & Error $(\epsilon_0,\epsilon_1)$ & Inputs & Function & Growth \\
\midrule
\cite{ranklower} & $\log_2\mathrm{SR}(\psi)$ & $(0,c)$ & $\forall$ & $\mathrm{EQ},\mathrm{GT},\mathrm{DISJ}$ & $\Omega(n)$ \\
\cite{cdsrouting} & $\log_2\mathrm{SR}(\psi)$ & $(0,c)$ or $(c,0)$ & $\forall$ & $\mathrm{IP}$ & $\Omega(n)$ \\
\cite{bcs} & qubits & $(c,c)$ & $\mathbb E$ & $\mathrm{IP}$ & $\Omega(\log n)$ \\
\cite{bcs,parallelrep} & qubits & $(c,c)$ & $\mathbb E$ & random $f$ & $\Omega(n)$ \\
\cite{gatebounds} & gates & $(c,c)$ & $\mathbb E$ & $\mathrm{IP}$ & $\Omega(n/\log q)$ \\
\midrule
This paper & $E_F(\rho_{LR})$ & $(c,c)$ & $\mathbb E$ & $\mathrm{IP}$ & $\Omega(\log n)$ \\
\bottomrule
\end{tabular}
\end{table}

To our knowledge, \Cref{thm:formation} is the first growing lower bound for an explicit function that charges only the shared state and allows error in both routing cases. The bounds that allow error in both routing cases charge every qubit Alice and Bob hold or their gates, or restrict their computation, so they do not isolate the shared entanglement~\cite{bcs,parallelrep,gatebounds,computational}. The related inner-product bounds in~\cite{relating} also rely on a qubit count that includes local memory. Earlier lower bounds on the entanglement of formation concern non-local implementations of bipartite unitaries~\cite{gonzaleschitambar,controllable}, and \Cref{thm:formation} is the first for $f$-routing that grows with the classical input length.

\subsection{Bob's state reveals the routing case}

\begin{figure}[tb]
\centering
\begin{tikzpicture}[thick,x=0.45cm,y=0.45cm,
 lab/.style={font=\small},
 bit/.style={font=\scriptsize,inner sep=1.5pt}]
\foreach \i in {0,...,7}{
 \foreach \j in {0,...,7}{
  \pgfmathtruncatemacro\f{mod(floor(\i/4)*floor(\j/4)
   +mod(floor(\i/2),2)*mod(floor(\j/2),2)+mod(\i,2)*mod(\j,2),2)}
  \ifnum\f=1
   \fill[gray!45] (\j,-\i) rectangle ++(1,-1);
  \fi
 }
}
\draw[thin] (0,0) grid[step=1] (8,-8);
\draw (0,0) rectangle (8,-8);
\foreach \i/\b in {0/000,1/001,2/010,3/011,4/100,5/101,6/110,7/111}{
 \node[bit,anchor=east] at (-0.05,-\i-0.5) {\b};
 \node[bit,anchor=west,rotate=90] at (\i+0.5,0.05) {\b};
}
\node[lab] at (-1.6,-4) {$x$};
\node[lab] at (4,1.8) {$y$};
\node[lab,anchor=west] at (9.2,-2.2) {For small recovery error on $(x,y)$};
\draw[fill=white] (9.4,-3.9) rectangle ++(0.8,0.8);
\node[lab,anchor=west] at (10.5,-3.5) {Alice recovers, $\chi(x,y)<1$};
\draw[fill=gray!45] (9.4,-5.2) rectangle ++(0.8,0.8);
\node[lab,anchor=west] at (10.5,-4.8) {Bob recovers, $\chi(x,y)>1$};
\end{tikzpicture}
\caption{The routing pattern of inner product for $n=3$. Alice must recover $Q$ where $f(x,y)=0$ (white cells), and Bob where $f(x,y)=1$ (gray cells). On each input pair with error at most $\epsilon<0.24$, $\chi(x,y)$ lies below $1$ in a white cell and above $1$ in a gray cell, by a constant margin.}\label{fig:pattern}
\end{figure}

The proof compares Bob's two possible states just after the exchange, when the qubit input is $\lvert0\rangle$ or $\lvert1\rangle$. If Bob can recover the qubit, he can distinguish these states reliably. If Alice can recover it, Bob can learn little about the input from Alice's message and the register he kept. The proof has five stages.
\begin{enumerate}
\item \emph{Bob's state defines a matrix} (\Cref{sec:model,sec:gap}).
The rows $x$ and columns $y$ each range over the $N=2^n$ possible input strings. The entry $\chi(x,y)$ measures how well Bob can tell from his registers whether Alice received $\lvert0\rangle$ or $\lvert1\rangle$.
\item \emph{A threshold separates the routing cases} (\Cref{sec:gap} and \Cref{fig:pattern}).
On an input pair where Alice recovers the qubit with small error, Bob can hardly distinguish the two states, and $\chi(x,y)$ lies below one. Where Bob recovers with small error, he can distinguish them reliably, and $\chi(x,y)$ lies above one. We call this separation the \emph{robust routing gap}.
\item \emph{An approximation with small trace norm} (\Cref{sec:rank}).
We approximate the matrix of $\chi$ because its trace norm is hard to bound. When the pure shared state has small Schmidt rank, the approximation has small trace norm. Each entry changes by at most $0.01$, too little to erase the routing gap on average when the recovery error is small.
\item \emph{Inner product forces large Schmidt rank} (\Cref{subsec:average-rank}).
We now choose inner product modulo $2$ as the routing function $f$. Changing Alice's string $x$ to any other string switches which party must recover the qubit for exactly half of Bob's strings $y$. A matrix that follows this pattern with a constant average margin above or below one must have large trace norm. By the previous step, the pure shared state must therefore have large Schmidt rank.
\item \emph{From Schmidt rank to entanglement of formation} (\Cref{subsec:entropy}).
Write the shared state as a mixture of pure states. In each component, discard the smallest Schmidt terms until its rank is below the bound from the previous step. Every component then fails the routing test too often, so their mixture does too. To cause this drop in success, the discarded terms must carry substantial weight on average. Spreading that weight over so many small Schmidt coefficients forces the average entanglement of every decomposition to grow at least logarithmically in $n$.
\end{enumerate}

\section{Routing model}\label{sec:model}

All Hilbert spaces are finite-dimensional. Fix a Boolean function $f:\{0,1\}^n\times\{0,1\}^n\to\{0,1\}$. Alice receives an $n$-bit string $x$ and a qubit $Q$, and Bob receives an $n$-bit string $y$. Before the inputs arrive, they share a resource state $\rho_{LR}$, with Alice holding $L$ and Bob holding $R$. They prepare the resource independently of both strings and of the joint state of $Q$ and any reference system. \Cref{fig:protocol} shows the exchange.

\begin{figure}[t]
\centering
\begin{tikzpicture}[thick,
 gate/.style={draw,fill=gray!12,rounded corners=2pt,minimum height=0.75cm,minimum width=2.2cm},
 lab/.style={font=\small}]
\draw[decorate,decoration={snake,amplitude=0.5mm,segment length=2.5mm}]
 (2.0,5.1) -- (7.0,5.1);
\node[lab,above] at (4.5,5.25) {$\rho_{LR}$};
\draw (0.6,5.5) -- (0.6,4.075);
\node[lab,left] at (0.6,5.1) {$Q$};
\draw[double,double distance=1pt] (1.3,5.5) -- (1.3,4.075);
\node[lab,right] at (1.3,5.1) {$x$};
\draw (2.0,5.1) -- (2.0,4.075);
\node[lab,left] at (2.0,4.5) {$L$};
\draw[double,double distance=1pt] (7.7,5.5) -- (7.7,4.075);
\node[lab,right] at (7.7,5.1) {$y$};
\draw (7.0,5.1) -- (7.0,4.075);
\node[lab,right] at (7.0,4.5) {$R$};
\node[gate] at (1.3,3.7) {$\mathcal A_x$};
\node[gate] at (7.7,3.7) {$\mathcal B_y$};
\draw (0.6,3.325) -- (0.6,1.175);
\node[lab,left] at (0.6,2.4) {$A$};
\draw (8.4,3.325) -- (8.4,1.175);
\node[lab,right] at (8.4,2.4) {$Y$};
\draw (2.0,3.325) -- node[lab,above right,pos=0.22] {$X$} (7.0,1.8) -- (7.0,1.175);
\draw (7.0,3.325) -- node[lab,above left,pos=0.22] {$B$} (2.0,1.8) -- (2.0,1.175);
\node[gate] at (1.3,0.8) {$\mathcal D_A^{x,y}$};
\node[gate] at (7.7,0.8) {$\mathcal D_B^{x,y}$};
\draw (1.3,0.425) -- (1.3,-0.15) node[lab,below] {$Q$};
\draw (7.7,0.425) -- (7.7,-0.15) node[lab,below] {$Q$};
\node[lab] at (1.3,5.95) {Alice};
\node[lab] at (7.7,5.95) {Bob};
\end{tikzpicture}
\caption{One-round $f$-routing, with time running downward. Double lines denote classical inputs, and the wavy line denotes the resource $\rho_{LR}$. Alice recovers the qubit from $AB$ when $f(x,y)=0$, and Bob recovers it from $XY$ when $f(x,y)=1$.}\label{fig:protocol}
\end{figure}

Alice applies a quantum channel $\mathcal A_x:QL\to AX$, keeps $A$ and sends $X$ to Bob. Bob applies a quantum channel $\mathcal B_y:R\to YB$, keeps $Y$ and sends $B$ to Alice. These channels allow arbitrary local operations. The messages can be classical or quantum, with no bound on their size. The exchange is simultaneous, so neither party sees the other's message before preparing its own. Alice then holds $AB$ and Bob holds $XY$. Bob's registers $X$ and $Y$ are named after the inputs $x$ and $y$ they depend on, and Alice's registers $A$ and $B$ after the parties that produced them.

The party selected by $f(x,y)$ must recover $Q$ without further communication. Alice applies $\mathcal D_A^{x,y}:AB\to Q$ when $f(x,y)=0$, and Bob applies $\mathcal D_B^{x,y}:XY\to Q$ when $f(x,y)=1$. These recovery channels may depend on both inputs because each party can include its classical input in its message. With the resource $\rho_{LR}$ fixed, let $\mathcal T_{x,y}$ be the channel from the input qubit to the qubit recovered by the selected party. The other party's registers are discarded.

The error of the protocol on the input pair $(x,y)$ compares the recovered qubit with the initial qubit $Q$. Formally, it is the distance between $\mathcal T_{x,y}$ and the identity channel $\id_Q$ in the diamond norm,
\begin{equation}\label{eq:diamond}
 \norm{\mathcal T_{x,y}-\id_Q}_{\diam}
 =\max_{\sigma}\norm{(\id_{\bar Q}\otimes\mathcal T_{x,y})(\sigma)-\sigma}_1.
\end{equation}
The maximum runs over all states $\sigma$ of $Q$ and a reference qubit $\bar Q$ that the protocol does not touch. Comparing their joint state before and after recovery also tests whether the protocol preserves entanglement with the reference. We use the full diamond norm~\cite{maynotes,complexity}, so the error ranges from $0$ for perfect recovery to $2$ if some joint state becomes orthogonal to its input. We count every input pair equally when computing the average error.

The test in \Cref{thm:formation} starts with $\bar QQ$ in the Bell state $\lvert\Phi\rangle=(\lvert00\rangle+\lvert11\rangle)/\sqrt2$. After recovery, the verifiers check whether the two qubits are still in this state. With $\Phi=\ketbra\Phi\Phi$ denoting its projector, the entanglement infidelity~\cite{schumacher}
\begin{equation}\label{eq:bell-infidelity}
 \delta_{x,y}=1-\Tr\bigl[\Phi(\id_{\bar Q}\otimes\mathcal T_{x,y})(\Phi)\bigr]
\end{equation}
is the probability that the test fails on input pair $(x,y)$. It is zero exactly when $\mathcal T_{x,y}$ is the identity channel. The average entanglement infidelity is
\begin{equation}\label{eq:average-infidelity}
 a=\mathbb E_{x,y}\delta_{x,y},
\end{equation}
where $\mathbb E_{x,y}$ means the average over all input pairs $(x,y)$, with each pair counted equally. \Cref{thm:formation} assumes that this average failure probability is below $0.12$. The Bell test never fails on $\Phi$, so its failure probability after recovery is at most the trace distance of the recovered joint state from $\Phi$~\cite{fuchsvandegraaf}. Since $\Phi$ is among the inputs in \eqref{eq:diamond}, this distance is at most half the diamond-norm error. Thus $a\le\epsilon/2$ whenever the average error is at most $\epsilon$. This includes protocols whose error is at most $\epsilon$ on every pair.

Bob's two states just before recovery show what happens in the two perfect routing cases. If Bob can recover $Q$ perfectly, he can tell the inputs $\lvert0\rangle$ and $\lvert1\rangle$ apart with certainty, so his two states are orthogonal. If Alice can recover $Q$ perfectly, Bob's registers carry no information about the qubit, so his two states are equal.

To describe Bob's state, we need only the outputs $X$ and $Y$ that he holds after the exchange. Fix the qubit input of Alice's channel $\mathcal A_x$ to $\lvert q\rangle$, with $q\in\{0,1\}$. Tracing out $A$ gives the channel $\mathcal E_x^q:L\to X$ for the message Bob receives. Tracing out $B$ from the output of Bob's channel $\mathcal B_y$ gives $\mathcal F_y:R\to Y$ for the register he keeps. Bob's state on $XY$ just before recovery is therefore
\begin{equation}\label{eq:bob-state}
 \rho^q_{x,y}=(\mathcal E_x^q\otimes\mathcal F_y)(\rho_{LR}).
\end{equation}
Alice's input $x$ enters only through $\mathcal E_x^q$, and Bob's input $y$ only through $\mathcal F_y$. The same channels give Bob's states for every choice of resource. In formulas for a single input pair, we drop the index $x,y$ and write, for example, $\rho^q$ and $\delta$.

\Cref{sec:gap} compares $\rho^0_{x,y}$ and $\rho^1_{x,y}$ to obtain the entry $\chi(x,y)$ of an $N\times N$ matrix.

\section{Robust routing gap}\label{sec:gap}

To compare Bob's two states, imagine that a verifier measures the reference qubit in the computational basis and keeps the outcome private. The verifier records $+1$ for the qubit input $\lvert0\rangle$ and $-1$ for $\lvert1\rangle$. Both signs are equally likely for every input pair $(x,y)$. The statistic $\chi(x,y)$ measures how well Bob can predict the sign by measuring his own registers $XY$.

Consider an input pair $(x,y)$ on which the recovery error is small. If $f(x,y)=0$, Alice's recovered qubit and the reference are close to a Bell pair before the verifier measures. A perfect Bell pair is independent of Bob's registers. Therefore, outcomes that let Bob confidently predict the sign must be rare. If $f(x,y)=1$, Bob can measure his recovered qubit in the same basis as the verifier and usually gets the same sign. This is why $\chi(x,y)$ lies below one in Alice's case and above one in Bob's.

\subsection{The routing statistic}\label{subsec:statistic}

Fix an input pair $(x,y)$ and let $Z$ be the sign the verifier obtains by measuring the reference qubit $\bar Q$ in the Pauli $\mathsf Z$ basis. Bob holds $\rho^0$ when $Z=+1$ and $\rho^1$ when $Z=-1$. Bob can predict any real number $b$, with squared error $(Z-b)^2$. Before he measures $XY$, both signs are equally likely, so his mean squared error is
\begin{equation}
 \tfrac12(1-b)^2+\tfrac12(-1-b)^2=1+b^2.
\end{equation}
Predicting zero minimizes this error. It means Bob favors neither sign. Zero is one unit away from either sign, so its squared error is always one. After measuring $XY$, Bob can base his prediction on the outcome. His best measurement lowers the mean squared error by $\chi(x,y)/2$.

Fix a measurement $\mathcal M$ on $XY$. For each outcome $m$ that can occur, let $p_m$ be its probability. The conditional mean $z_m=\mathbb E[Z\mid m]$ is the average sign in runs where Bob sees $m$. Predicting $b$ on this outcome gives mean squared error $1-z_m^2+(b-z_m)^2$. The best prediction is therefore $z_m$, which lowers the error from one to $1-z_m^2$. Averaging this gain over the outcomes and choosing the best measurement gives
\begin{equation}\label{eq:chi}
 \frac{\chi(x,y)}2=\sup_{\mathcal M}\sum_m p_m z_m^2 .
\end{equation}
If $p_m^0$ and $p_m^1$ are the probabilities of $m$ on the two states, then $p_m=(p_m^0+p_m^1)/2$ and $p_m z_m=(p_m^0-p_m^1)/2$. For identical states Bob learns nothing and $\chi=0$. For orthogonal states he can read the sign without error and $\chi=2$. How much recovery error can we allow while keeping $\chi$ below one for Alice and above one for Bob?

\subsection{Robustness to recovery errors}\label{subsec:direct-gap}

The gap survives whenever the entanglement infidelity $\delta_{x,y}$ is below $1/8$, because $\chi(x,y)$ stays within $8\delta_{x,y}$ of its value for perfect recovery.

\begin{proposition}[Robust routing gap]\label{prop:bell-gap}
For all input pairs $(x,y)$ and resource states, $0\le\chi(x,y)\le2$ and
\begin{equation}\label{eq:bell-sld-gap}
 f(x,y)=0\Longrightarrow\chi(x,y)\le8\delta_{x,y},
 \qquad
 f(x,y)=1\Longrightarrow\chi(x,y)\ge2-8\delta_{x,y}.
\end{equation}
\end{proposition}

\begin{proof}
Fix an input pair and write $\delta=\delta_{x,y}$. Since $\abs{z_m}\le1$, the prediction formula \eqref{eq:chi} gives $0\le\chi\le2$.

\begin{list}{\labelitemi}{%
  \setlength{\leftmargin}{1em}%
  \setlength{\labelwidth}{0.5em}%
  \setlength{\labelsep}{0.5em}}
\item \emph{Alice recovers, $f(x,y)=0$.} Leave the reference qubit $\bar Q$ unmeasured and fix a measurement for Bob. Alice recovers $Q$ from $AB$, and Bob measures $XY$. For each outcome $m$ of Bob's measurement, let $\omega_m$ be the joint state of the reference qubit $\bar Q$ and Alice's recovered qubit $Q$. Neither party acts on $\bar Q$, so the verifier's sign still has conditional mean $z_m$ on $\omega_m$. Its mean on the Bell state $\Phi$ is zero. Since the sign is $+1$ or $-1$, we have $\abs{z_m}\le\norm{\omega_m-\Phi}_1$. The fidelity bound for a pure state~\cite{fuchsvandegraaf} thus gives $z_m^2\le4(1-\Tr(\Phi\omega_m))$. Weighting by the probability $p_m$ of each outcome gives
\begin{equation}
 \begin{aligned}
 \sum_m p_m z_m^2
 &\le 4\sum_m p_m\bigl(1-\Tr(\Phi\omega_m)\bigr)\\
 &=4\delta.
 \end{aligned}
\end{equation}
Bob measures only $XY$, so averaging over his outcomes gives the original Bell-test failure probability $\delta$. The bound holds for every measurement Bob can choose, so \eqref{eq:chi} gives $\chi/2\le4\delta$ and hence $\chi\le8\delta$.

\item \emph{Bob recovers, $f(x,y)=1$.} Bob recovers $Q$ from $XY$, measures it in the same basis as the verifier and uses his sign as his prediction. Disagreement corresponds to the two basis states $\lvert01\rangle$ and $\lvert10\rangle$. Both are orthogonal to $\lvert\Phi\rangle$, so the probability of disagreement is at most the Bell-test failure probability $\delta$. Bob's squared error is four when the signs disagree and zero otherwise. Its average is therefore at most $4\delta$, giving $\chi\ge2-8\delta$.\qedhere
\end{list}
\end{proof}

\section{Approximation with small trace norm}\label{sec:rank}

The gap between the two routing cases survives a small change in each entry of the matrix $\chi$. We use this freedom to construct an approximation $A$ that will help us lower bound the Schmidt rank of the resource. The trace norm $\norm A_1$ is the sum of the singular values of $A$~\cite{nuclearnorm}. For a diagonal matrix, the singular values are the absolute values of its diagonal entries. For example, $\operatorname{diag}(2,1)$ has trace norm $2+1=3$. Here we show that a small Schmidt rank lets us choose $A$ with small trace norm. We write $\norm{\chi-A}_{\max}$ for the largest absolute error in any entry.

\begin{proposition}[Approximation of $\chi$ with small trace norm]\label{prop:chi-rank}
For all integers $d\ge1$, every protocol using a pure resource of Schmidt rank at most $d$ has a real $N\times N$ approximation matrix $A$ satisfying
\begin{equation}
 \norm{\chi-A}_{\max}\le0.01,
 \qquad
 \norm A_1\le(70d^4)^{20d}N,
\end{equation}
where $\chi$ denotes the matrix with entries $\chi(x,y)$.
\end{proposition}

We introduce $A$ here to state what we need from it before we construct it. Each entry $A_{x,y}$ must be within $0.01$ of $\chi(x,y)$. We also need the matrix to have small trace norm when the Schmidt rank $d$ is small. We finish the proof of \Cref{prop:chi-rank} at the end of \Cref{subsec:moments}.

\subsection{A product reference for Bob's states}\label{subsec:product-terms}

Just before recovery, Bob holds Alice's message $X$ and the register $Y$ he kept. Their joint state is $\rho^0_{x,y}$ for qubit input $\lvert0\rangle$ and $\rho^1_{x,y}$ for $\lvert1\rangle$. For a pure resource of Schmidt rank $d$, we count how many product terms we need to describe each state. We use this count in \Cref{subsec:moments} to bound the trace norm.

The Schmidt decomposition of the shared state $\lvert\psi\rangle$ uses $d$ basis vectors on Alice's side and $d$ on Bob's. On each side these vectors span a $d$-dimensional space called the Schmidt support. With Schmidt coefficients $s_r$, the shared state and its density matrix are
\begin{equation}\label{eq:compatible-output}
 \begin{aligned}
 \lvert\psi\rangle
 &=\sum_{r=1}^d s_r\lvert u_r\rangle\otimes\lvert v_r\rangle,\\
 \rho_{LR}=\ketbra\psi\psi
 &=\sum_{r,t=1}^d s_r s_t\,\ketbra{u_r}{u_t}\otimes\ketbra{v_r}{v_t}.
 \end{aligned}
\end{equation}
Each term in the ket pairs with each term in the bra, giving $d^2$ product terms. Substituting this expansion of $\rho_{LR}$ into the formula for Bob's state~\eqref{eq:bob-state} applies $\mathcal E_x^q$ to the first factor of each term and $\mathcal F_y$ to the second~\cite{ranklower}. Alice receives the input qubit $\lvert q\rangle$, so changing $q$ can affect only the first factor. Bob's average state $\tau=(\rho^0+\rho^1)/2$ and half-difference $\Delta=(\rho^0-\rho^1)/2$ therefore each have at most $d^2$ product terms.

To bound the error in approximating $\chi$, we compare Bob's average state with a product state that keeps $x$ and $y$ separate. For this comparison, suppose Alice and Bob start with independent maximally mixed states $I_L/d$ and $I_R/d$ on their Schmidt supports. After the exchange, Bob's average state in this case is $\Omega_{x,y}=\alpha_x\otimes\beta_y$, where
\begin{equation}
 \alpha_x=\overline{\mathcal E}_x(I_L/d),\qquad
 \beta_y=\mathcal F_y(I_R/d),\qquad
 \overline{\mathcal E}_x=(\mathcal E_x^0+\mathcal E_x^1)/2.
\end{equation}
Here $\alpha_x$ is the state of Alice's message $X$, averaged over the two qubit inputs, and $\beta_y$ is the state of the register $Y$ Bob kept. The lemma shows that any measurement outcome is at most $d^2$ times as likely on Bob's actual average state as on this product state.

\begin{lemma}[Comparison with the product state]\label{lem:source-domination}
  For every input pair, Bob's average state satisfies $\tau_{x,y}\preceq d^2\Omega_{x,y}$.
\end{lemma}

\begin{proof}
Multiplying the comparison input $(I_L/d)\otimes(I_R/d)$ by $d^2$ gives the identity on the joint Schmidt support. This identity upper bounds the shared state $\ketbra\psi\psi$. Applying the same channel $\overline{\mathcal E}_x\otimes\mathcal F_y$ to both sides preserves the inequality and gives $\tau_{x,y}\preceq d^2\Omega_{x,y}$.
\end{proof}

We next use this comparison to approximate every entry of $\chi$ to within $0.01$. The Schmidt rank gives an upper bound on how many terms we need for this accuracy.

\subsection{The statistic as a combination of moments}\label{subsec:quadratic}

Each entry $\chi(x,y)$ tells us how well Bob can predict the sign for that input pair. Bounding the trace norm of the whole matrix $\chi$ directly is hard. Therefore, we approximate each entry by a weighted sum of quantities called moments. Across all input pairs, each moment gives a matrix whose trace norm we can bound. We combine these estimates to bound the trace norm of the approximation $A$.

Fix an input pair $(x,y)$. Bob holds $\rho^0$ for qubit input $\lvert0\rangle$ and $\rho^1$ for $\lvert1\rangle$. Their average is $\tau=(\rho^0+\rho^1)/2$ and their half-difference is $\Delta=(\rho^0-\rho^1)/2$. The product state $\Omega=\alpha_x\otimes\beta_y$ is what Bob holds on average when Alice and Bob start with independent maximally mixed states on their Schmidt supports.

For a state $\sigma$ and any matrix $H$ on Bob's registers $XY$, write $\mathcal L_\sigma(H)=(\sigma H+H\sigma)/2$. This averages the two ways of multiplying $H$ by $\sigma$, from the left and from the right. To form the moments, we use $\mathcal L_\Omega^{-1}$ to divide out the weights from the product state $\Omega$. In an eigenbasis of $\Omega$, this divides each matrix entry by the average of its row and column eigenvalues. We only need the eigenvectors with positive eigenvalues, since their span contains both of Bob's states by \Cref{lem:source-domination}.

To compute the moments, we repeat the same update $\mathsf T=\mathcal L_\Omega^{-1}\mathcal L_\tau$ on a sequence of matrices $K_j$.
\begin{equation}\label{eq:moment-def}
 \begin{aligned}
 K_0 &= \mathcal L_\Omega^{-1}(\Delta),\\
 K_{j+1} &= \mathsf T K_j,\\
 m_j(x,y) &= \Tr(\Delta K_j),\qquad j\ge0.
 \end{aligned}
\end{equation}
The index $j$ counts the updates. The lemma gives one set of coefficients for all input pairs. We bound the sum of their absolute values to control the trace norm of the approximation matrix $A$.

\begin{lemma}[Approximating $\chi$ by a weighted sum of moments]\label{lem:moment-reduction}
There are real coefficients $c_0,\dots,c_{20d-2}$, depending only on $d$, with $\sum_j\abs{c_j}\le7^{20d}$, such that every protocol using a pure resource of Schmidt rank $d$ satisfies, for all input pairs,
\begin{equation}\label{eq:polynomial-error}
 \Bigl|\chi(x,y)-2\sum_{j=0}^{20d-2}c_jm_j(x,y)\Bigr|\le0.01.
\end{equation}
\end{lemma}

\begin{proof}
We first express Bob's prediction gain and the moments as sums over the same eigenvalues. For a fixed input pair $(x,y)$, Liu's formula for measured quantum triangular discrimination~\cite{liu} gives Bob's optimal gain in \eqref{eq:chi} as
\begin{equation}\label{eq:optimal-gain}
 \frac{\chi(x,y)}2=\Tr(\Delta H_*),\qquad
 \mathcal L_\tau(H_*)=\Delta.
\end{equation}
Bob attains this gain by measuring the observable $H_*$ on $XY$ and reporting the eigenvalue he obtains. By $\mathcal L_\tau(H_*)=\Delta$ in \eqref{eq:optimal-gain}, each reported eigenvalue lies in $[-1,1]$, so we can choose $H_*$ with $\norm{H_*}_{\mathrm{op}}\le1$.

On the support of $\Omega$, use the inner product $\inner HK_\Omega=\Tr(H\mathcal L_\Omega(K))$ for Hermitian matrices. Then $\inner H{\mathsf TK}_\Omega=\Tr(H\mathcal L_\tau(K))$ is symmetric in $H$ and $K$, so the update $\mathsf T$ is self-adjoint. For an eigenvector $H$, the eigenvalue is $\Tr(\tau H^2)/\Tr(\Omega H^2)$, which lies in $[0,d^2]$ by \Cref{lem:source-domination}.

Write $\mu_i$ for the eigenvalues of $\mathsf T$ and $w_i$ for the squared coordinates of $H_*$ in an orthonormal eigenbasis. Substituting $\Delta=\mathcal L_\tau(H_*)$ from \eqref{eq:optimal-gain} into the definition of $K_0$ in \eqref{eq:moment-def} gives $K_0=\mathsf T H_*$. Each update applies $\mathsf T$ once more, so $K_j=\mathsf T^{j+1}H_*$. By the definition of the inner product, $\Tr(\Delta K)=\inner{\mathsf T H_*}K_\Omega$. Taking $K=H_*$ for Bob's gain in \eqref{eq:optimal-gain} and $K=K_j$ for the moment $m_j$ in \eqref{eq:moment-def} gives
\begin{equation}\label{eq:spectral-moments}
 \frac{\chi(x,y)}2=\sum_i w_i\mu_i,
 \qquad m_j(x,y)=\sum_i w_i\mu_i^{j+2}.
\end{equation}
We therefore seek a polynomial that starts at $\mu^2$ and approximates $\mu$ on $[0,d^2]$.

Put $k=20d$ and $R(\mu)=T_k(1-2\mu/d^2)$, where $T_k$ is the Chebyshev polynomial of the first kind~\cite{dlmf}. Since $T_k(1)=1$ and $T'_k(1)=k^2$, we have $R(0)=1$ and $R'(0)=-800$. The constant and linear terms therefore cancel in
\begin{equation}
 P(\mu)=\mu-\frac{1-R(\mu)}{800}
       =\sum_{j=0}^{k-2}c_j\mu^{j+2}.
\end{equation}
Replacing each $\mu_i$ by $P(\mu_i)$ in the sum for $\chi/2$ gives $\sum_j c_jm_j$. Since $\sum_i w_i=\Tr(\Omega H_*^2)\le1$ and $\abs{R}\le1$ on $[0,d^2]$, doubling this sum approximates $\chi$ with error at most $1/200$. This is within our tolerance of $0.01$.

All roots of $R$ are positive, so its coefficients alternate in sign and their absolute values sum to $R(-1)$. The formula for $P$ then gives $\sum_j\abs{c_j}\le R(-1)$, and the Chebyshev formula gives $R(-1)=T_k(1+2/d^2)\le7^k$.
\end{proof}

We now have an approximation matrix $A$ that is within $0.01$ of $\chi$ in every entry. We form $A$ by multiplying each moment matrix $m_j$ by its weight $2c_j$ and adding the results entry by entry. To bound the trace norm of $A$ in \Cref{prop:chi-rank}, we next bound the trace norm of each moment matrix in terms of the Schmidt rank $d$ of the shared state.

\subsection{Trace norm of the moments}\label{subsec:moments}

To bound the trace norm of each moment, we separate its dependence on $x$ and $y$. A Fourier integral~\cite{hiaikosaki} lets us write the moment as an integral of sums of products. We can then count the products and bound their size.

\begin{lemma}[Trace norm of a moment]\label{cor:moment-resource}
For all $j\ge0$ and every protocol using a pure resource of Schmidt rank $d$, the matrix of the moment $m_j$ satisfies
\begin{equation}
 \norm{m_j}_1
 \le\Bigl(\frac{j+2}2\Bigr)^{j+1}d^{3j+6}N.
\end{equation}
\end{lemma}

Appendix~\ref{app:moment-rank-proof} proves \Cref{cor:moment-resource}. We now have everything we need to prove \Cref{prop:chi-rank}.

\begin{proof}[Proof of \Cref{prop:chi-rank}]
Set $A=2\sum_j c_jm_j$ with the coefficients from \Cref{lem:moment-reduction}. By \eqref{eq:polynomial-error}, the error in each entry is at most $0.01$. For the trace norm, both $2$ and $(j+2)/2$ are at most $10d$ because $j+2\le20d$. Thus \Cref{cor:moment-resource} gives $2\norm{m_j}_1\le N(10d^4)^{20d}$ for each moment in the sum. Combining this with the coefficient bound from \Cref{lem:moment-reduction} gives
\begin{equation}
 \begin{aligned}
 \norm A_1 &\le \sum_j \abs{c_j}\,2\norm{m_j}_1\\
 &\le 7^{20d}N(10d^4)^{20d}\\
 &= N(70d^4)^{20d}.
 \end{aligned}
\end{equation}
For Schmidt rank below $d$, the bound only decreases.
\end{proof}

For any routing function $f$, the Schmidt rank of the pure shared state gives an upper bound on the trace norm of the approximation matrix $A$. We next choose $f$ to be inner product modulo $2$ and use its routing pattern to lower bound the same trace norm. Comparing the two bounds tells us how large the Schmidt rank must be to route inner product.


\section{Entanglement lower bound}\label{sec:formation}

We first show that a mixture of pure states of small Schmidt rank cannot route inner product with small average error as $n$ grows. We then use this bound to prove \Cref{thm:formation} for mixed resources.

\subsection{Inner-product routing requires large Schmidt rank}\label{subsec:average-rank}

How far does each entry of the approximation matrix $A$ lie on the correct side of one? Subtract one from each entry and write $B_{xy}=A_{xy}-1$. Give each input pair the sign $S_{xy}=+1$ if Bob must recover and $S_{xy}=-1$ if Alice must recover. Then $S_{xy}B_{xy}$ counts the margin above one for Bob and below one for Alice. An entry on the wrong side counts negatively.

For inner product, any two distinct rows of $S$ disagree in exactly half the columns, so their dot product is zero. Dividing each row by its length $\sqrt N$ gives an orthogonal matrix~\cite{kushilevitznisan}. By trace-norm duality, the sum of the entrywise products of $B$ with any orthogonal matrix is at most $\norm B_1$~\cite{nuclearnorm}. Choosing $S/\sqrt N$ gives
\begin{equation}\label{eq:walsh-correlation}
 \norm B_1 \ge \sum_{x,y}\frac{S_{xy}}{\sqrt N}\,B_{xy}
 = N^{3/2}\,\mathbb E_{x,y}[S_{xy}B_{xy}].
\end{equation}
Thus the trace norm of $B$ is at least $N^{3/2}$ times the average routing margin.

Fix the local channels and the recovery channels. For any resource $\sigma$ on $L\otimes R$, write $a(\sigma)=\mathbb E_{x,y}\delta_{x,y}(\sigma)$ for its average entanglement infidelity. A state has Schmidt number at most $d$ if we can write it as a mixture of pure states with at most $d$ Schmidt terms each~\cite{schmidtnumber}.

\begin{lemma}[Schmidt number under average error]\label{lem:bounded-rank-ensemble}
Suppose a protocol for $f_n$ uses a resource $\sigma$ of Schmidt number $d$ and has average entanglement infidelity $a(\sigma)\le0.12$. Then $n\le168\,d\log_2(3d)$, and for $n\ge2$,
\begin{equation}
 \log_2d\ge\log_2n-\log_2\log_2n-9.
\end{equation}
\end{lemma}

\begin{proof}
Decompose $\sigma$ into pure states of Schmidt rank at most $d$. The failure probability $a(\sigma)$ is a weighted average over these components, so some component $\phi$ has $a(\phi)\le0.12$. For this component, the routing gap in \Cref{prop:bell-gap} gives $\chi$ an average margin of at least $1-8a(\phi)\ge 1-8(0.12)=0.04$. Approximating $\chi$ by $A$ can reduce the margin by at most $0.01$ in each entry, so the average margin of $A$ is at least $0.04-0.01=0.03$.

The margin bound~\eqref{eq:walsh-correlation} tells us how large $\norm B_1$ must be. For an upper bound, recall that $B$ is $A$ minus the all-ones matrix, whose trace norm is $N$. By \Cref{prop:chi-rank}, rounding $70d^4$ up to $(3d)^4$ gives $\norm A_1\le N(3d)^{80d}$. Comparing the bounds gives
\begin{equation}
 \begin{aligned}
 0.03 N^{3/2} &\le \norm B_1\\
 &\le \norm A_1+N\\
 &\le 2N(3d)^{80d}.
 \end{aligned}
\end{equation}
Divide by $0.03N$. Since $2/0.03\le(3d)^{4d}$, we get $2^{n/2}=\sqrt N\le(3d)^{84d}$. Taking logarithms gives $n\le168d\log_2(3d)$.

For $n\ge2$, we can replace $\log_2(3d)$ by $\log_2 n$ in the first bound. This only makes the right-hand side smaller when $3d>n$, and then $168d$ already exceeds $n$ while $\log_2 n\ge1$. Taking logarithms of $n\le168d\log_2 n$ and using $\log_2 168<9$ gives the second bound.
\end{proof}

Thus routing inner product with small average error requires a Schmidt number that grows at least as fast as $n/\log_2 n$. We next use this to bound the entanglement of formation.

\subsection{From Schmidt rank to entanglement of formation}\label{subsec:entropy}

\begin{proof}[Proof of \Cref{thm:formation}]
Fix the protocol channels. Let $K$ be the smallest Schmidt number of a resource on $L\otimes R$ with average failure probability at most $0.12$. The shared state $\rho_{LR}$ meets this threshold, and \Cref{lem:bounded-rank-ensemble} lower bounds $\log_2 K$. Changing the resource changes the failure probability by at most the trace distance~\cite{fuchsvandegraaf}. Thus every state within trace distance $0.12-a$ of $\rho_{LR}$ also meets the threshold and has Schmidt number at least $K$.

Choose a decomposition of $\rho_{LR}$ whose average entanglement equals $E_F(\rho_{LR})$. In each pure state, order the Schmidt terms by weight (squared Schmidt coefficient), largest first. The terms from position $K$ onward must together carry at least $(0.12-a)^2$ of the weight on average. Otherwise, discarding them and normalizing each component would leave a mixture of Schmidt number below $K$. Its trace distance from $\rho_{LR}$ would be less than $0.12-a$, because this distance is at most the square root of the average discarded weight~\cite{dilution,fuchsvandegraaf}. Such a mixture would still meet the routing threshold, contradicting the choice of $K$.

The first $K$ Schmidt weights sum to at most one, so no weight from position $K$ onward exceeds $1/K$. Each such weight $\lambda$ contributes $\lambda\log_2(1/\lambda)\ge\lambda\log_2 K$ to the entropy. Together these terms carry at least $(0.12-a)^2$ of the weight on average, so they alone give
\begin{equation}
 E_F(\rho_{LR})\ge (0.12-a)^2\log_2 K.
\end{equation}
The lower bound on $\log_2 K$ from \Cref{lem:bounded-rank-ensemble} now proves the theorem.
\end{proof}

\section{Conclusion}\label{sec:conclusion}

This paper shows that $f$-routing for inner product with any constant error below $0.24$, on average over uniformly random input pairs, needs a shared state whose entanglement of formation grows at least logarithmically in $n$ (\Cref{thm:formation}). The bound holds for arbitrary mixed shared states, messages of any size and any local operations. To our knowledge, it is the first robust lower bound that grows with $n$, holds for an explicit function and charges only the shared state; the other known growing bounds of this kind need perfect recovery in one routing case~\cite{ranklower,cdsrouting}.

Our logarithmic lower bound for inner product is the strongest our current method can prove. For the upper bound, $2n+2$ shared EPR pairs suffice to route every Boolean function perfectly~\cite{linearupper}. In the model of \Cref{sec:model}, whether a robust lower bound on the shared entanglement can grow faster than logarithmically, and ultimately polynomially in $n$, remains open~\cite{maynotes}.

\appendix
\numberwithin{theorem}{section}
\aliascntresetthe{proposition}
\aliascntresetthe{lemma}
\renewcommand{\theHproposition}{\thesection.\arabic{proposition}}
\renewcommand{\theHlemma}{\thesection.\arabic{lemma}}

\section{Trace norm of a moment}\label{app:moment-rank-proof}

\begin{proof}[Proof of \Cref{cor:moment-resource}]
Expanding a moment gives a trace of $\ell=j+2$ operators, with one denominator for each of the $\ell-1$ inverses. A Fourier integral separates each denominator across $x$ and $y$. We distribute the eigenvalue powers among those denominators so that one application of H\"older's inequality bounds the whole trace.

\subsection{The denominators in a moment}

Fix $(x,y)$ and work on $\supp\Omega$. Write the average and half-difference as
\begin{equation}
 \tau=\Omega^{1/2}C\Omega^{1/2},\qquad
 \Delta=\Omega^{1/2}D\Omega^{1/2}.
\end{equation}
Both $C$ and $D$ are sums of at most $d^2$ product terms. Normalizing the product terms from \Cref{subsec:product-terms} by $\Omega^{-1/2}=\alpha_x^{-1/2}\otimes\beta_y^{-1/2}$ preserves each product. Both operator norms are at most $d^2$ by $-\tau\preceq\Delta\preceq\tau\preceq d^2\Omega$.

Expand the $j$ symmetric multiplications in \eqref{eq:moment-def} by choosing left or right multiplication by $\tau$, each with weight $1/2$. For a fixed choice, expand the trace in an eigenbasis of $\Omega$. The matrix entries of the $j$ copies of $C$ and two copies of $D$ form a cycle with $\ell$ index positions. Write $\nu_v$ for the eigenvalue at position $v$. Each position receives two factors $\sqrt{\nu_v}$ from the neighboring operators. Each inverse $\mathcal L_\Omega^{-1}$ divides an entry by the mean of the eigenvalues at its two indices. Connecting these two positions by an edge for each inverse gives a tree $T$. Indeed, the innermost inverse connects the two indices of $D$. Each multiplication on the left or right introduces a new index on that side, and the next inverse connects it to the index kept on the other side. The scalar weight multiplying the cyclic product of entries is therefore
\begin{equation}\label{eq:moment-weight}
 \prod_{v}\nu_v\;\prod_{uv\in T}\frac{2}{\nu_u+\nu_v}.
\end{equation}

\subsection{Separating the inputs}\label{app:moment-rank-final}

We want to leave a power $1/\ell$ of each eigenvalue for the trace. Root the tree at a uniformly chosen vertex and charge each edge to its child. Every vertex spends one power except the root, which keeps its power. Averaging these exponents leaves $1/\ell$ at every vertex. An edge $uv$ draws a fraction $\theta$ from $u$ and $1-\theta$ from $v$, where $\theta$ is the probability that the root lies on $v$'s side of the edge.

For $a=\nu_u$ and $b=\nu_v$, we separate the denominator using the Mellin integral~\cite{dlmf}
\begin{equation}\label{eq:mellin-inverse}
 \frac{2a^\theta b^{1-\theta}}{a+b}
 =\int_{\mathbb R}(a/b)^{it}\,w_\theta(t)\,dt,
 \qquad w_\theta(t)=\frac{1}{\sin(\pi(\theta-it))}.
\end{equation}
The real part of the sine in the denominator is $\sin(\pi\theta)\cosh(\pi t)$. Since $1/\cosh(\pi t)$ integrates to one and $1/\ell\le\theta\le1-1/\ell$, the total absolute weight is at most $1/\sin(\pi\theta)\le\ell/2$ by concavity of sine on $[0,\pi]$.

At fixed frequencies, the factors $(a/b)^{it}$ add imaginary powers to the eigenvalue at each position. Summing over the indices restores the trace, with one power of $\Omega$ between consecutive operators,
\begin{equation}
 F(x,y)=\Tr\left(\Omega^{1/\ell+it_1}O_1\cdots
                      \Omega^{1/\ell+it_\ell}O_\ell\right).
\end{equation}
Here the $O_i$ are the two copies of $D$ and the $j$ copies of $C$, and each $t_i$ is the signed sum of the frequencies on the incident edges. Each power has Schatten $\ell$-norm $(\Tr\Omega)^{1/\ell}=1$. Since $\norm{O_i}_{\mathrm{op}}\le d^2$, H\"older's inequality gives $\abs{F(x,y)}\le d^{2\ell}$. To bound the rank of $F$, expand the $O_i$ into product terms. Each power of $\Omega$ also factors across $X$ and $Y$, so each resulting trace is a product $u(x)v(y)$. There are at most $d^{2\ell}$ such products. Cauchy--Schwarz on the singular values therefore gives
\begin{equation}
 \begin{aligned}
 \norm F_1 &\le\sqrt{\rank F}\,\Bigl(\sum_{x,y}\abs{F(x,y)}^2\Bigr)^{1/2}\\
 &\le N d^{3\ell}.
 \end{aligned}
\end{equation}
The choices of sides and all Fourier weights are independent of $(x,y)$, so we can integrate the bound for the whole matrix by the triangle inequality. Each of the $\ell-1$ edges contributes a factor of at most $\ell/2$. Averaging over the choices of sides gives
\begin{equation}
 \norm{m_j}_1\le\Bigl(\frac{\ell}{2}\Bigr)^{\ell-1}N d^{3\ell}.
\end{equation}
Substituting $\ell=j+2$ gives the bound in \Cref{cor:moment-resource}.
\end{proof}

\section*{Lean formalization}

A Lean formalization proves \Cref{thm:formation} using only Lean's three standard axioms. The code is available at \url{https://github.com/kevinbogner/f-routing-lean}.

\section*{Author contributions and AI use}

The author directed the research. GPT-6.2 Astra (OpenAI) was the main AI assistant for the mathematics, Lean formalization and manuscript revisions. Claude Fable 5.2 (Anthropic) helped edit the text and express \Cref{thm:formation} in terms of average entanglement infidelity. The author takes responsibility for the paper's content.

\providecommand{\etalchar}[1]{$^{#1}$}

\end{document}